\RequirePackage{amsmath}

\documentclass[runningheads]{llncs}

\usepackage[T1]{fontenc}
\usepackage{lmodern}

\usepackage{amssymb}
\usepackage{amsmath}
\usepackage{float}

\usepackage{tikz}
\tikzset{every picture/.style={line width=0.75pt}}

\begin{document}

\title{Fast FPT Algorithms for Grundy Number on Dense Graphs} 

\author{Sina Ghasemi Nezhad\inst{1}\orcidID{0009-0001-9747-2121} \and Maryam Moghaddas\inst{1}\orcidID{0009-0008-1611-1633} \and Fahad Panolan\inst{2}\orcidID{0000-0001-6213-8687}}

\authorrunning{S. Ghasemi Nezhad, M. Moghaddas, and F. Panolan.}

\institute{Department of Mathematical Sciences, Sharif University of Technology, Tehran, Iran \\ \email{\{sina.ghaseminejad, maryam.moghaddas\}@sharif.edu} \and School of Computer Science, University of Leeds, Leeds, UK \\ \email{f.panolan@leeds.ac.uk}}

\maketitle

\begin{abstract}
    In this paper, we investigate the \textsc{Grundy Coloring} problem for graphs with a cluster modulator, a structure commonly found in dense graphs. The Grundy chromatic number, representing the maximum number of colors needed for the first-fit coloring of a graph in the worst-case vertex ordering, is known to be $W[1]$-hard when parameterized by the number of colors required by the most adversarial ordering. We focus on fixed-parameter tractable (FPT) algorithms for solving this problem on graph classes characterized by dense substructures, specifically those with a cluster modulator. A cluster modulator is a vertex subset whose removal results in a cluster graph (a disjoint union of cliques). We present FPT algorithms for graphs where the cluster graph consists of one, two, or $k$ cliques, leveraging the cluster modulator’s properties to achieve the best-known FPT runtimes, parameterized by both the modulator’s size and the number of cliques.

\end{abstract}

\section{Introduction}

The \textit{first-fit} coloring is a heuristic that assigns to each vertex, arriving in a specified order $\sigma$, the smallest available color such that the coloring remains proper. Here, we denote the colors as natural numbers $\{1,2,3,\ldots\}$. The {\em Grundy chromatic number} (or simply {\em Grundy number}) is the number of colors that are needed for the most adversarial vertex ordering $\sigma$, i.e., the maximum number of colors that the first-fit coloring requires over all possible vertex orderings. In other words, if the Grundy number of a graph $G$ is $\ell$, then $\ell$ is the largest integer such that the following holds. There is an ordered partition $(V_1, \ldots, V_{\ell})$ of the vertex set $V(G)$ of $G$  such that for all $i\in [\ell]$ and $v\in V_i$, $N_G(v)\cap V_j\neq \emptyset$ for all $j\in \{1,\ldots, i-1\}$ and $V_i$ is an independent set. Here, we use $N_G(v)$ to denote the set of neighbors of $v$ and $[\ell]$ to denote the set $\{1,2,\ldots,\ell\}$. Also, we say that any ordered partition $(U_1, \ldots, U_{\ell^\prime})$ of $V(G)$ is a {\em Grundy coloring} of $G$ if it satisfies the above-mentioned property. That is, for all $i\in [\ell^\prime]$ and $v\in U_i$, $N_G(v)\cap U_j\neq \emptyset$ for all $j \in \{1, \ldots, i-1\}$ and $U_i$ is an independent set. The Grundy number was first introduced by Patrick Michael Grundy but was formally defined by Christen and Selkow \cite{christen1979some}.

In the \textsc{Grundy Coloring} problem, we are given a graph $G$ and an integer $\ell$, and we want to test whether the Grundy number of $G$ is at least $\ell$ or not. The \textsc{Grundy Coloring} problem has been shown to be NP-hard even on special graph classes like bipartite graphs, their complements, chordal graphs, and line graphs, as established by known results in the literature \cite{zaker2005grundy,sampaio2012algorithmic,havet2013grundy,havet2015complexity}. While this makes solving the problem in general infeasible in polynomial time, efficient algorithms have been designed for special graph classes. The \textsc{Grundy Coloring} for a tree can be solved in linear time \cite{hedetniemi1982linear}, and this result has been generalized with a polynomial time algorithm for the bounded-treewidth graphs \cite{telle1997algorithms}.

\textsc{Grundy Coloring} is well studied in the realm of parameterized complexity. Aboulker et al. \cite{aboulker2023grundy} proved that \textsc{Grundy Coloring} is $W[1]$-hard when parameterized by the number of colors required by the most adversarial vertex ordering, leveraging the structure of half-graphs. Furthermore, their work also establishes that the brute-force algorithm, with a running time of $f(\ell) n^{2^{\ell-1}}$, is essentially optimal under the Exponential Time Hypothesis (ETH). In the context of parameterized complexity, the problem has been studied with various parameters. For instance, Bonnet et al. \cite{bonnet2018complexity} explored the complexity of the  \textsc{Grundy Coloring} problem, identifying a range of parameterizations for which the problem remains computationally challenging. Similarly, Belmonte et al. \cite{belmonte2022grundy} considered the parameters treewidth and pathwidth, and proved the following surprising result. \textsc{Grundy Coloring} is $W[1]$-hard when parameterized by treewidth, but FPT when parameterized by pathwidth. Their work also demonstrates that \textsc{Grundy Coloring} is FPT when parameterized by the neighborhood diversity (and more generally modular width). The running time of their algorithm is $2^{O(w2^w)}n^{O(1)}$, where $w$ is the modular width or neighborhood diversity of the input graph.  

We study \textsc{Grundy Coloring} problem on dense graphs from the perspective of parameterized complexity. We consider the problem when parameterized by $k$-cluster modulator. That is, the input consists of a graph $G$ and a vertex subset $R\subseteq V(G)$ of size $r$ such that $G-R$ is a $k$-cluster graph (i.e., $G-R$ has $k$ connected components and each of them is a complete graph). Here, the parameter is $k+r$. Observe that in this case, the neighborhood diversity of $G$ is at most $k 2^{r} + r$. Thus, the result of  Belmonte et al. \cite{belmonte2022grundy} implies that there is an algorithm with a running time of $2^{O(2^{k2^r})}n^{O(1)}$. Notice that even when $k=1$, this running time is at least $2^{\Omega(2^{2^r})}$.

Our main results are when $k=1$ and $k=2$. When $k=1$, the parameter is the clique modulator. In this case, we design an algorithm with a running time of $O(2^{O(r^2)} + n + m)$, where $m$ is the number of edges in the graph, and this result can be found in Section~\ref{sec:clique}. 

As a byproduct of our proof, we also get that the problem admits a kernel of size $O(r2^r)$.  When $k=2$, we design an algorithm with a running time of $O(2^{O(r^2)}n^6)$. To prove this result we model the problem as many instances of the classic max flow problem. This result can be found in Section~\ref{sec:flow}. 

For general $k$ we give faster algorithm with a running time of $2^{O(2^{kr})}n^{O(1)}$ in Section~\ref{sec:kcluster}. To prove this result we first observe some structural results on Grundy coloring and then model the problem as an Integer Linear Program with few variables, like the method used by Belmonte et al. \cite{belmonte2022grundy}.

\section{Preliminaries} \label{sec:prelims}

For  a natural number $q\in \{1,2,\ldots\}$, we use $[q]$ to denote the set $\{1,2\ldots,q\}$.  Let $a = (a_1, \ldots, a_n)$ and $b = (b_1, \ldots, b_m)$ be two sequences. The \textit{concatenation} of $a$ and $b$, denoted $ab$, is the sequence formed by appending all elements of $b$ to the end of $a$, resulting in $(a_1, \ldots, a_n, b_1, \ldots, b_m)$. A sequence $c = (c_1, \ldots, c_s)$ is called a \textit{subsequence} of a sequence $d = (d_1, \ldots, d_t)$ if there exists a strictly increasing sequence of indices $1 \leq i_1 < \cdots < i_s \leq t$ such that $c_j = d_{i_j}$ for all $j \in \{1, \ldots, s\}$.

We use standard graph notations in this paper: for a graph $G$, the set of vertices is denoted by $V(G)$ (with a size of $|V(G)| = n$), and the set of edges is denoted by $E(G)$ (with a size of $|E(G)| = m$). For a vertex $v \in V(G)$, we use $N_G (v)$ to denote the open neighborhood set of $v$ in $G$. That is, $N_G (v)=\{u\in V(G) : \{u,v\}\in E(G)\}$. We also use  $N_G [v]$ for the closed neighborhood set of $v$ in $G$, i.e., $N_G (v) \cup \{v\}$. For a vertex subset $X \subseteq V$, $G[X]$ represents the induced subgraph of $G$ over $X$. That is, $V(G[X])=X$ and $E(G[X])=\{\{u,v\}\in E(G) : u,v\in X\}$. For $X\subseteq V(G)$, $G-X$ denotes the induced subgraph $G[V(G)\setminus X]$.

A vertex subset $I\subseteq V(G)$ is called an {\em independent set} if for any pair of vertices $u,v\in X$, $\{u,v\}\notin E(G)$. A vertex subset $K\subseteq V(G)$ is called an {\em clique} if for any pair of vertices $u,v\in X$, $\{u,v\}\in E(G)$ (i.e., $G[K]$ is a complete graph). A graph $G$ is a {\em cluster graph} if each connected component of $G$ is a complete graph. We say that $G$ is a {\em $k$-cluster graph} if $G$ is a cluster graph and the number of connected components in $G$ is $k$. For a graph $G$, a vertex subset $R\subseteq V(G)$ is a {\em clique modulator} if $V(G)\setminus R$ is a clique in $G$.  For a graph $G$, a vertex subset $R\subseteq V(G)$ is a {\em cluster modulator} (respectively, {\em $k$-cluster modulator}) if $G-R$ is a cluster graph (respectively, $k$-cluster graph).

A graph coloring is called \textit{proper} if no two adjacent vertices share the same color; otherwise, it is referred to as \textit{improper}. In the context of proper colorings, we define \textit{color classes} as subsets $C \subseteq V(G)$, where all vertices within $C$ are assigned the same color in the coloring of graph $G$. We denote the color classes of $G$ as $(C_1, \ldots, C_\gamma)$, where $\gamma$ represents the total number of colors used.

Now we prove some simple results about Grundy coloring which we use in the later sections.

\begin{lemma} \label{lemma:swap-vertices} 
    Let $G$ be a graph and  $u,v\in V(G)$ such that $N_G[u] = N_{G}[v]$. Let $\sigma$ be a permutation of $V(G)$ and $\hat{\sigma}$ is the permutation obatined from $\sigma$ by swapping the vertices $u$ and $v$ in $\sigma$. Then, the first-fit colorings of $G$ with respect to $\sigma$ and $\hat{\sigma}$ use the same number of colors.   
\end{lemma}

\begin{proof}
    Assume that in $\sigma$, vertex $v$ appears before $u$, and in $\hat{\sigma}$, vertex $u$ appears before $v$. For any vertex $x$ that appears before $u$ in $\hat{\sigma}$, the color assigned to $x$ will not change, as the neighbors of $x$ that appear earlier than $x$ in the permutation remain unchanged. Note that the vertices appearing before $v$ in $\sigma$ are exactly the same as those appearing before $u$ in $\hat{\sigma}$. Therefore, since $N_G[u] = N_G[v]$, vertex $u$ will receive the same color in $\hat{\sigma}$ as $v$ does in $\sigma$.

    Next, consider any vertex $y$ that appears between $u$ and $v$ in $\hat{\sigma}$. If $y \notin N_G[u]$ (or equivalently $y \notin N_G[v]$), then swapping $u$ and $v$ will not affect the colors assigned to the neighbors of $y$ that appear before $y$ in the permutation. If $y \in N_G[u]$, when we reach $y$ in the first-fit coloring of $\hat{\sigma}$, the set of occupied colors is unchanged since the only modification is that $u$ appears before $y$ instead of $v$. As shown earlier, $u$ receives the same color in $\hat{\sigma}$ as $v$ does in $\sigma$, so the color of $y$ remains the same as in $\sigma$. Thus, similar to the argument made for $v$, since $N_G[u] = N_G[v]$, vertex $v$ will receive the same color in $\hat{\sigma}$ as $u$ does in $\sigma$.

    Finally, for any vertex $z$ that appears after $v$ in $\hat{\sigma}$, the colors of its neighbors that appear earlier than $z$ remain unchanged, except for the possible swap of colors between $u$ and $v$. Therefore, the set of occupied colors remains the same, and $z$ will receive the same color as in $\sigma$, which completes the proof. \qed
\end{proof}

\begin{lemma} \label{lemma:perm-to-colorclass} 
    Consider a permutation $\sigma$ of the vertices and perform the first-fit coloring to it to reach the color classes $(C_1, \ldots, C_\gamma)$. For each $1 \leq i \leq \gamma$, let the vertices in $C_i$ be $\{v_{i1}, \ldots, v_{i n_i}\}$ where $n_i$ is the number of vertices in the $i^\text{th}$ color class. Then the following permutation is called the sorted permutation of $\sigma$ and will have the exact same color classes as $\sigma$ under the first-fit coloring:
    $$
        \sigma_s = v_{11}, v_{12}, \ldots, v_{1 n_1}, v_{21}, v_{22}, \ldots, v_{2 n_2}, \ldots, v_{\gamma 1}, v_{\gamma 2}, \ldots, v_{\gamma n_\gamma}
    $$
\end{lemma}

\begin{proof}
    Assume that by performing the first-fit coloring, there exists a vertex $ u $ such that it is the first element in $\sigma_s$ that receives a different color class from the one it had in $\sigma$. Let the color class of $ u $ in $\sigma$ be $ C_x $. 

    First, suppose that the index of the color class of $ u $ in $\sigma_s$ is less than its index in $\sigma$. This implies that in $\sigma$, there are vertices that belong to the color classes $ C_1, \ldots, C_{x-1} $ that appear before $ u $ and are neighbors of $ u $. These vertices will also appear before $ u $ in $\sigma_s$, as $\sigma_s$ orders vertices within each color class while preserving the order of the color classes themselves. Since these vertices are neighbors of $ u $, the first-fit coloring of $\sigma_s$ cannot assign $ u $ a color corresponding to a class earlier than $ C_x $, contradicting the assumption that $ u $'s color class index in $\sigma_s$ is lower than in $\sigma$.

    Next, suppose that the index of the color class of $u$ in $\sigma_s$ is greater than its index in $\sigma$. In this case, if $u$ appears later in $\sigma_s$ than expected, there must be a vertex $v \in N_G(u)$ that appears before $u$ in $\sigma_s$ and belongs to the color class $C_x$. Since $u$ is assumed to be the first vertex in $\sigma_s$ with a different color class from its class in $\sigma$, $v$ must also belong to $C_x$ in $\sigma$. This implies that $u$ and $v$, being neighbors, would share the same color class in $\sigma$, which contradicts the definition of proper coloring.

    Since both possibilities lead to contradictions, it follows that the sorted permutation $\sigma_s$ will produce the exact same color classes as $\sigma$ under the first-fit coloring. Therefore, the lemma is proven. \qed
\end{proof}

Due to the above lemma, throughout this paper, we will present a color classification instead of a permutation of vertices to represent a Grundy coloring.

\begin{lemma} \label{lemma:singletons-at-the-end-G}
    Let $G$ be a graph, and $\gamma$ be the Grundy number of $G$. There is a Grundy coloring $(C_1,\ldots,C_{\gamma})$ of $G$ with the following property. There is an index $i \in \{1,\ldots,\gamma\}$ such that for each $j \leq i$, we have $|C_j| > 1$, and for each $j>i$, we have $|C_j|=1$.
\end{lemma}

\begin{proof}
    The lemma states that we can rearrange the singleton color classes in a Grundy coloring of $G$ such that all singleton color classes appear at the end, resulting in a proper coloring that still uses the maximum number of colors.
    
    Consider a permutation $\sigma$ such that the first-fit coloring of $G$ with respect to $\sigma$ uses $\gamma$ colors. Let $(C_1, \ldots, C_\gamma)$ be the corresponding color classes. If this coloring satisfies the property stated in the lemma, then we are done. Suppose, that this is not the case. Then, there must exist indices $x$ and $w$ such that $w > x$, $|C_x| = 1$, and $|C_w| \geq 2$. Since $C_x$ is a singleton, the vertex $u \in C_x$ must be adjacent to at least one vertex in each of the color classes $C_y$ for all $y < x$. Additionally, the vertices in $C_z$ for all $z > x$ must also be adjacent to $u$. This implies that $u$ is adjacent to a vertex in each color class except $C_x$. Therefore, $(C_1, \ldots, C_{x-1}, C_{x+1}, \ldots, C_\gamma, C_x)$ is a Grundy coloring of $G$.
 
    By repeating this process for all singleton classes, we obtain a proper color classification in which all singleton classes appear at the end, which is the desired result and completes the proof. \qed    
\end{proof}

\section{Grundy Coloring with a Clique Modulator} \label{sec:clique}

In this section, we design a faster FPT algorithm for \textsc{Grundy Coloring} when parameterized by the size of the given clique modulator. Here, the input is a graph $G$ and a clique modulator $R$ of $G$ of size $r$.  Our main idea is to construct a permutation of vertices in $G$, which results in a maximum number of colors using first-fit coloring. 
We say that a Grundy coloring of $G$ is an optimal solution if the number of colors used is the Grundy number of $G$. 
We start describing the properties of an optimal solution with the following lemmas.

\begin{lemma} \label{lemma:singletons-at-end}
    Let $G$ be a graph, $R$ be a clique modulator of $G$, and $\gamma$ be the Grundy number of $G$. There is a Grundy coloring $(C_1,\ldots,C_{\gamma})$ of $G$ with the following property. There is an index $i\in \{1,\ldots,\gamma\}$ such that for each $j\leq i$, $C_j\cap R \neq \emptyset$, and for each $j>i$, $C_{j}\cap R =\emptyset$ and $|C_j|=1$.
\end{lemma}

\begin{proof}
     The proof follows a similar approach to that of Lemma~\ref{lemma:singletons-at-the-end-G}. Note that at most one vertex from the clique can be in any color class. By applying the same method, but focusing only on the color classes $C_{j}$ of size one that satisfies $C_{j}\cap R =\emptyset$, instead of considering all color classes $C_{j}$ of size one, we arrive at the desired result. \qed
\end{proof}

\begin{lemma} \label{lemma:arbitaryorderingsuffix}
    Let $G$ be a graph, $R$ be a clique modulator of size at most $r$, and $\gamma$ be the Grundy number of $G$. There is a vertex subset $Q\subseteq V(G)\setminus R$ of size $r$ with the following property. Let $\sigma_2$ be an arbitrary ordering of $V(G)\setminus (R\cup Q)$. Then, there exists an ordering $\sigma_1$ of $Q\cup R$ such that the first-fit coloring of $G$ with respect to $\sigma_1 \sigma_2$ uses $\gamma$ number of colors.
\end{lemma}

\begin{proof}
    We begin by applying Lemma~\ref{lemma:singletons-at-end}, which guarantees the existence of a Grundy coloring $(C_1, \ldots, C_{\gamma})$ of $G$ such that there exists an index $i \in \{1, \ldots, \gamma\}$ with the following properties: For each $j \leq i$, $C_j \cap R \neq \emptyset$, and for each $j > i$, $C_j \cap R = \emptyset$ and $|C_j| = 1$. This implies that all singleton color classes, consisting of vertices from $V(G) \setminus R$, appear at the end of the coloring. Since each color class $C_j$ for $j \leq i$ contains at most one vertex from $V(G) \setminus R$, we define $Q = \bigcup_{j = 1}^r (C_j \setminus R)$. By construction, $Q$ contains at most $r$ vertices from $V(G) \setminus R$ within the first $r$ color classes, ensuring that $|Q| \leq r$.

    Observe that we have $i \leq r$. All remaining color classes $(C_{r+1}, \ldots, C_{\gamma})$ are singletons where $C_j \cap R = \emptyset$ for each $j \in \{r+1, \ldots, \gamma\}$. Therefore, we can define an arbitrary ordering $\sigma_2$ for the vertices in $V(G) \setminus (R \cup Q)$, corresponding to the color classes $(C_{r+1}, \ldots, C_{\gamma})$. We then construct the ordering $\sigma_1$ for the vertices in $Q \cup R$ based on the color classes $(C_1, \ldots, C_r)$, ensuring that it respects the first-fit coloring that resulted in the color classes $(C_1, \ldots, C_\gamma)$. Thus, by concatenating $\sigma_1$ and $\sigma_2$, the first-fit coloring of $G$ with respect to this ordering uses exactly $\gamma$ colors, completing the proof. \qed
\end{proof}

Let $S=V(G)\setminus R$. Notice that $S$ is a clique in $G$. Now we define an equivalence relation $\sim_R$ on $S$ as follows. For any two vertices $u,v\in S$, $u\sim_R v$ if and only if $N_G[u]=N_G[v]$. Let $E_1,\ldots,E_q$ be the equivalence classes of $\sim_R$. It is easy to see that $q\leq 2^{r}$. Now from each equivalence class $E_i$ arbitrarily select a subset $F_i\subseteq E_i$ of size $\min\{r,|E_i|\}$. Let $F=\bigcup_{i=1}^q F_i$. Constructing $F$ is efficient. For each vertex $v$, we determine its equivalence class by examining its neighbors in $R$. The vertex $v$ is included in $F$ only if its equivalence class within $F$ contains fewer than $r$ vertices. This procedure involves iterating over all vertices and the edges between $R$ and $S$, resulting in a time complexity of $O(n + m)$.

Now we prove the following lemma.  

\begin{lemma} \label{lemma:kernelclique}
    The Grundy number of $G$ is equal to the sum of the Grundy number of $G[R \cup F]$ and $|V(G) \setminus (R \cup F)|$. 
\end{lemma}

\begin{proof}
    The proof is based on the following observation. If $\sigma_2$ is an arbitrary ordering of $V(G)\setminus (R\cup F)$. Then for any ordering $\sigma_1$ of $G[R\cup F]$, the number of colors in the first-fit coloring of $G$ with respect to $\sigma_1\sigma_2$ is equal to the sum of the number of colors in the first-fit coloring of $G$ with respect to $\sigma_1$, and $|V(G)\setminus (R\cup F)|$.
    
    According to Lemma~\ref{lemma:singletons-at-end}, there exists a color classification $(C_1, \ldots, C_\gamma)$ using Grundy number of colors such that there is an index $i \in \{1,\ldots,\gamma\}$ where, for each $j \leq i$, we have $C_j \cap R \neq \emptyset$, and for each $j > i$, we have $C_j \cap R = \emptyset$ and $|C_j| = 1$. Since $C_j \cap R = \emptyset$ for all $j > i$, and $R$ contains $r$ vertices, we have $i \leq r$.

    Consider the first $r^\prime$ color classes, $(C_1, \ldots, C_{r^\prime})$, containing exactly $r$ vertices from the clique $S$. Since there are at least $r$ vertices in $S$, we have $r^\prime \geq r$. Using Lemma~\ref{lemma:swap-vertices}, we can swap any vertex from $S$ in these color classes with a vertex from the set $F$. Given that $F$ includes at least $\min\{r,|E_i|\}$ vertices from each equivalence class of $\sim_R$, this swap is feasible. After the swap, the vertices from $S$ in the first $r^\prime$ color classes are exactly those from $F$, and their number is $r$. Let $Q \subseteq F$ denote these $r$ vertices from $F$.

    The modified color classification $(C_1, \ldots, C_{r^\prime})$ now includes vertices from $R \cup Q$ in the first $r^\prime$ color classes, with the remaining color classes being singletons that can be ordered at the end. The vertices in $F \setminus Q$ can be placed as singleton color classes before the singleton color classes from $V(G) \setminus (R \cup F)$. Finally, with $|V(G) \setminus (R \cup F)|$ singleton vertices left, these can be assigned in any order at the end. Thus, the Grundy number of $G$ equals the sum of the Grundy number of $G[R \cup F]$, and $|V(G) \setminus (R \cup F)|$, as the first $r^\prime$ color classes include all non-singleton vertices from $R \cup Q \subseteq R \cup F$, and the remaining vertices are assigned their own colors. 

    The set $Q$ satisfies the necessary conditions of Lemma~\ref{lemma:arbitaryorderingsuffix}. Therefore, if $\sigma^\prime_2$ is an arbitrary ordering of $V(G) \setminus (R \cup Q)$, then there exists an ordering $\sigma^\prime_1$ of $Q \cup R$ such that the first-fit coloring of $G$ with respect to the concatenated order $\sigma^\prime_1 \sigma^\prime_2$ uses exactly $\gamma$ colors. To leverage this result, we choose $\sigma^\prime_2$ as an ordering in which the vertices of $F \setminus Q$ appear first. We now define $\sigma_1$ to be the concatenation of $\sigma^\prime_1$ with the first $|F \setminus Q|$ vertices from $\sigma^\prime_2$, effectively placing these vertices immediately after the vertices in $R \cup Q$. Finally, we define $\sigma_2$ as the ordering of the remaining vertices in $\sigma^\prime_2$. This construction ensures that the orderings $\sigma_1$ and $\sigma_2$ satisfy the desired conditions, and the proof is complete. \qed
\end{proof}

The following theorem follows from Lemma~\ref{lemma:kernelclique}. 

\begin{theorem}
    The \textsc{Grundy Coloring} problem parameterized by the size $r$ of a clique modulator admits a kernel of size $O(r2^{r})$.
\end{theorem}

\begin{theorem} \label{theorem:clique-modulator-fpt}
    There is an algorithm that given a graph $G$ and a clique modulator $R$ of size $r$, runs in time $O(2^{O(r^2)}+n+m))$ and outputs a Grundy coloring of $G$ using the maximum number of colors.
\end{theorem}

\begin{proof}
    Let $F\subseteq V(G)$ be the subset mentioned in Lemma~\ref{lemma:kernelclique}. Recall that $|F|\leq r 2^r$. Let $G^\prime=G[R\cup F]$. Moreover, by Lemma~\ref{lemma:kernelclique}, to prove the theorem it is enough to get an ordering $\sigma_1$ of $G^\prime$ such that first-fit coloring with respect to $\sigma_1$ uses the Grundy number of $G^\prime$ many colors. To get such an ordering $\sigma_1$, we use Lemma~\ref{lemma:arbitaryorderingsuffix}. By Lemma~\ref{lemma:arbitaryorderingsuffix}, there is a vertex subset $Q\subseteq V(G^\prime)\setminus R$ of size $r$ with the following property. Let $\sigma^\prime_2$ be an arbitrary ordering of $V(G^\prime)\setminus (R\cup Q)$. Then, there exists an ordering $\sigma^\prime_1$ of $Q\cup R$ such that the first-fit coloring of $G^\prime$ with respect to $\sigma^\prime_1 \sigma^\prime_2$ uses Grundy number of $G^\prime$ many colors.
    
    Since $V(G^\prime)\setminus R=F$ and $|F|\leq r2^r$, the number of choices for $Q$ is bounded by $\binom{r2^r}{r}$, which is upper bounded by $r^r2^{r^2}$. For each $Q$, the number of choices for $\sigma^\prime_1$ is upper bounded by $r^{2r}$. Thus trying all choices for $Q$ and $\sigma^\prime_1$ is upper bounded by $2^{O(r^2)}$. Thus in time $2^{O(r^2)}$, we can compute an ordering $\sigma_1$ such that the first fit coloring of $G^\prime$ with respect to $\sigma_1$ uses Grundy number of $G^\prime$ many colors. Since the construction of the set $F$ takes $O(n+m)$ time, the total running time is $O(2^{O(r^2)}+n+m)$. \qed
\end{proof}

\section{Grundy Coloring with a 2-Cluster Modulator} \label{sec:flow}

Next, we examine the parameter of a 2-cluster modulator. After
removing the 2-cluster modulator $R$ from $G$, the remaining graph $G - R$ forms a 2-cluster graph, meaning that $G - R$ consists of two connected components, $K_1$ and $K_2$, which are complete graphs. Let $S = V(K_1)$ and $S^\prime = V(K_2)$, where $S$ and $S^\prime$ are cliques in $G$. The key idea for this case is to reduce it to a max-flow problem. Let the size of the modulator be $|R| = r$. We begin with the following simple lemma.

\begin{lemma} \label{lemma:singletons-one-clique}
    Let $(C_1, \ldots, C_\gamma)$ be a Grundy coloring of $G$. Then, there do not exist two distinct color classes $C_i$ and $C_j$ such that $|C_i|=|C_j|=1$ and $C_i\subseteq S$ and $C_j\subseteq S^\prime$
\end{lemma}

\begin{proof}
    Assume that there exist color classes $C_i$ and $C_j$ such that $|C_i| = |C_j| = 1$ such that the element $x \in C_i$ is a vertex of $S$, and the element $y \in C_j$ is a vertex of $S^\prime$. Without loss of generality, assume that $1 \leq i < j \leq \gamma$. Therefore, since $y$ is colored with a number greater than $i$, due to the definition of first-fit coloring, it must have an edge to a vertex with color $i$, which can only be $x$. But since $x$ and $y$ are from two disjoint cliques with no edges between them, this is not possible. This contradiction completes the proof. \qed
\end{proof}

We now define an equivalence relation $\sim_R$ on the set $S \cup S^\prime$. Specifically, for any two vertices $u, v \in S \cup S^\prime$, we say $u \sim_R v$ if and only if they have the same closed neighborhood, i.e., $N_G[u] = N_G[v]$. The equivalence classes of $\sim_R$ are denoted $E_1, \ldots, E_q$ and $E^\prime_1, \ldots, E^\prime_{q^\prime}$, where $E_i$'s come from the clique $S$ and $E^\prime_j$'s come from the clique $S^\prime$. Similar to the case of a clique modulator, we have $q + q^\prime \leq 2^{r+1}$. Now from each equivalence class $E_i$ arbitrarily select a subset $F_i\subseteq E_i$ of size $\min\{r,|E_i|\}$. Also, for each equivalence class $E^\prime_j$ arbitrarily select a subset $F^\prime_j\subseteq E^\prime_j$ of size $\min\{r,|E^\prime_j|\}$. Let $F = \Big( \bigcup_{j} F_j \Big) \cup \Big( \bigcup_i F^\prime_i \Big)$. Notice that $|F|\leq r2^{r+1}$. 

\begin{lemma} \label{lemma:R_colorclasses-2}
    There is a vertex subset $Q\subseteq F$ of size at most $2r$ and a (not necessarily optimal) Grundy coloring $(C^\prime_1, \ldots, C^\prime_{\gamma^\prime})$ of $G[Q\cup R]$ with the following property. There is an optimal Grundy coloring $(C_1, \ldots, C_{\gamma})$ of $G$ such that $(C^\prime_1, \ldots, C^\prime_{\gamma^\prime})$ is a subsequence of $(C_1, \ldots, C_{\gamma})$.
\end{lemma}

\begin{proof}
    Consider an optimal Grundy coloring $(\widehat{C}_1, \ldots, \widehat{C}_\gamma)$ of $G$. Let $(\widehat{C}_{\alpha_1}, \ldots, \widehat{C}_{\alpha_{r^\prime}})$ be the color classes that contain at least one vertex from $R$. Clearly, $r^\prime \leq r$, as there are at most $r$ such color classes. Furthermore, note that each of the color classes $\widehat{C}_{\alpha_1}, \ldots, \widehat{C}_{\alpha_{r^\prime}}$ can contain at most two vertices from $V(G) \setminus R$, since each class contains at most one vertex from each clique in $G - R$. 
    
    By Lemma~\ref{lemma:swap-vertices}, we can swap the vertices of $(\widehat{C}_{\alpha_1}, \ldots, \widehat{C}_{\alpha_{r^\prime}})$ that belong to $V(G) \setminus (R\cup F)$ with vertices from $F$ that are in the same equivalence classes under $\sim_R$. This is feasible because $F$ contains $\min\{r,|E_i|\}$ and $\min\{r,|E_j^\prime|\}$ vertices from each equivalence class $E_i$ and $E_j^\prime$, respectively, and each color class contains at most one vertex from the clique $S$ or $S^\prime$, ensuring that the swap maintains the structure of the coloring.

    After performing these swaps, we obtain a new Grundy coloring $(C_1, \ldots, C_\gamma)$, and $(C_{\alpha_1}, \ldots, C_{\alpha_{r^\prime}})$ denote the color classes that contain the vertices from $R$. Moreover, for each $i\in [r^\prime]$, $C_{\alpha_i}\subseteq R\cup F$. Let  $Q= (C_{\alpha_1}\cup \ldots \cup  C_{\alpha_{r^\prime}})\cap F$. Notice that $|Q|\leq 2r$. Thus, $(C_1, \ldots, C_\gamma)$ is the desired optimal Grundy coloring. \qed 
\end{proof}

\begin{definition} \label{def:ext}
    A Grundy coloing $(C^\prime_1, \ldots, C^\prime_{\gamma^\prime})$ of $G[Q\cup R]$ is called {\em extendable} if there is an optimal Grundy coloring $(C_1, \ldots, C_{\gamma})$ of $G$ such that $(C_1^\prime, \ldots, C^\prime_{\gamma^\prime})$ is a subsequence of $(C_1, \ldots, C_{\gamma})$.
\end{definition}

\begin{lemma} \label{lemma:grundy-calc-2}
    Let $(C^\prime_1, \ldots, C^\prime_{\gamma^\prime})$ be a (not necessarily optimal) Grundy coloring of $G[Q\cup R]$, that is extendable. Let $\beta= \max\{|S\setminus Q|, |S^\prime\setminus Q|\}$. Then, Grundy number of $G$ is $\gamma^\prime+\beta$.    
\end{lemma}

\begin{proof}
    Since the sequence of color classes $(C^\prime_1, \ldots, C^\prime_{\gamma^\prime})$ is extendable, there exists an optimal Grundy coloring $(C_1, \ldots, C_{\gamma})$ of $G$ such that $(C^\prime_1, \ldots, C^\prime_{\gamma^\prime})$ is a subsequence of $(C_1, \ldots, C_{\gamma})$. Furthermore, note that the color classes which appear in $(C_1, \ldots, C_{\gamma})$ but not in $(C^\prime_1, \ldots, C^\prime_{\gamma^\prime})$ consist exclusively of vertices from $V(G) \setminus R$.

    By Lemma~\ref{lemma:singletons-one-clique}, these remaining color classes are either composed of two vertices (one from each clique) or a single vertex (from the clique with the greater number of remaining vertices). Hence, the total number of these additional color classes corresponds to $\beta = \max\{|S \setminus Q|, |S^\prime \setminus Q|\}$, where $S$ and $S^\prime$ are the vertex sets of the cliques in $G - R$. Consequently, the Grundy number of $G$ is given by $\gamma = \gamma^\prime + \beta$, where $\gamma^\prime$ represents the number of color classes in the subsequence $(C^\prime_1, \ldots, C^\prime_{\gamma^\prime})$. \qed
\end{proof}

Because of Lemma~\ref{lemma:R_colorclasses-2}, the number of choices for $Q$ is at most $2^{O(r^2)}$. For each such choice the number of Grundy colorings of $G[Q\cup R]$ is $2^{O(r \log r)}$ because $|Q\cup R|\leq 3r$. Thus, we can guess the correct choice for $Q$ and a Grundy coloring $(C^\prime_1, \ldots, C^\prime_{\gamma^\prime})$ of $G[Q\cup R]$, that is extendable in time $2^{O(r^2)}$. Thus our next job is to test whether $(C^\prime_1, \ldots, C^\prime_{\gamma^\prime})$ is indeed extendable. Towards that, we construct a flow network and prove that the network has a {\em large} flow if and only if $(C^\prime_1, \ldots, C^\prime_{\gamma^\prime})$ is extendable. 

\begin{theorem} \label{theorem:poly-extendable-2}
    There is an algorithm that given a graph $G$ and a Grundy coloring $(C^\prime_1, \ldots, C^\prime_{\gamma^\prime})$ of $G[Q\cup R]$, runs in time $O(n^6)$ and decides if $(C^\prime_1, \ldots, C^\prime_{\gamma^\prime})$ is extendable. 
\end{theorem}

\begin{proof}
    The algorithm we propose is based on a flow network. Without loss of generality, assume that $|S \setminus Q| \geq |S^\prime \setminus Q|$. Let the vertices in $S \setminus Q$ be $u_1, \ldots, u_s$, and the vertices in $S^\prime \setminus Q$ be $u_1^\prime, \ldots, u_{s^\prime}^\prime$. Our flow network consists of $s + 1$ main sections that form a bipartite subgraph, $s$ auxiliary vertices, a source $v_{\text{source}}$ that is adjacent to one part, and a sink $v_{\text{sink}}$ that is adjacent to the second part.

    \begin{figure}[t]
    \centering
        \scalebox{0.65}{
            \begin{tikzpicture}[x=0.75pt,y=0.75pt,yscale=-1,xscale=1]
                \draw   (80,309) .. controls (80,304.58) and (83.58,301) .. (88,301) -- (572,301) .. controls (576.42,301) and (580,304.58) .. (580,309) -- (580,333) .. controls (580,337.42) and (576.42,341) .. (572,341) -- (88,341) .. controls (83.58,341) and (80,337.42) .. (80,333) -- cycle ;
                \draw  [color={rgb, 255:red, 0; green, 0; blue, 0 }  ,draw opacity=1 ][fill={rgb, 255:red, 0; green, 0; blue, 0 }  ,fill opacity=1 ][line width=1.5]  (329.03,170.98) .. controls (329.03,170.44) and (329.47,170) .. (330.02,170) .. controls (330.56,170) and (331,170.44) .. (331,170.98) .. controls (331,171.53) and (330.56,171.97) .. (330.02,171.97) .. controls (329.47,171.97) and (329.03,171.53) .. (329.03,170.98) -- cycle ;
                \draw  [color={rgb, 255:red, 0; green, 0; blue, 0 }  ,draw opacity=1 ][fill={rgb, 255:red, 0; green, 0; blue, 0 }  ,fill opacity=1 ][line width=1.5]  (329.03,419.98) .. controls (329.03,419.44) and (329.47,419) .. (330.02,419) .. controls (330.56,419) and (331,419.44) .. (331,419.98) .. controls (331,420.53) and (330.56,420.97) .. (330.02,420.97) .. controls (329.47,420.97) and (329.03,420.53) .. (329.03,419.98) -- cycle ;
                \draw   (240,217.8) .. controls (240,213.39) and (243.58,209.8) .. (248,209.8) -- (412,209.8) .. controls (416.42,209.8) and (420,213.39) .. (420,217.8) -- (420,241.8) .. controls (420,246.22) and (416.42,249.8) .. (412,249.8) -- (248,249.8) .. controls (243.58,249.8) and (240,246.22) .. (240,241.8) -- cycle ;
                \draw  [color={rgb, 255:red, 0; green, 0; blue, 0 }  ,draw opacity=1 ][fill={rgb, 255:red, 0; green, 0; blue, 0 }  ,fill opacity=1 ][line width=1.5]  (144.03,360.98) .. controls (144.03,360.44) and (144.47,360) .. (145.02,360) .. controls (145.56,360) and (146,360.44) .. (146,360.98) .. controls (146,361.53) and (145.56,361.97) .. (145.02,361.97) .. controls (144.47,361.97) and (144.03,361.53) .. (144.03,360.98) -- cycle ;
                \draw  [color={rgb, 255:red, 0; green, 0; blue, 0 }  ,draw opacity=1 ][fill={rgb, 255:red, 0; green, 0; blue, 0 }  ,fill opacity=1 ][line width=1.5]  (274.03,359.98) .. controls (274.03,359.44) and (274.47,359) .. (275.02,359) .. controls (275.56,359) and (276,359.44) .. (276,359.98) .. controls (276,360.53) and (275.56,360.97) .. (275.02,360.97) .. controls (274.47,360.97) and (274.03,360.53) .. (274.03,359.98) -- cycle ;
                \draw  [color={rgb, 255:red, 0; green, 0; blue, 0 }  ,draw opacity=1 ][fill={rgb, 255:red, 0; green, 0; blue, 0 }  ,fill opacity=1 ][line width=1.5]  (514.03,359.98) .. controls (514.03,359.44) and (514.47,359) .. (515.02,359) .. controls (515.56,359) and (516,359.44) .. (516,359.98) .. controls (516,360.53) and (515.56,360.97) .. (515.02,360.97) .. controls (514.47,360.97) and (514.03,360.53) .. (514.03,359.98) -- cycle ;
                \draw  [color={rgb, 255:red, 0; green, 0; blue, 0 }  ,draw opacity=1 ][fill={rgb, 255:red, 0; green, 0; blue, 0 }  ,fill opacity=1 ][line width=1.5]  (269.03,229.98) .. controls (269.03,229.44) and (269.47,229) .. (270.02,229) .. controls (270.56,229) and (271,229.44) .. (271,229.98) .. controls (271,230.53) and (270.56,230.97) .. (270.02,230.97) .. controls (269.47,230.97) and (269.03,230.53) .. (269.03,229.98) -- cycle ;
                \draw  [color={rgb, 255:red, 0; green, 0; blue, 0 }  ,draw opacity=1 ][fill={rgb, 255:red, 0; green, 0; blue, 0 }  ,fill opacity=1 ][line width=1.5]  (289.03,229.98) .. controls (289.03,229.44) and (289.47,229) .. (290.02,229) .. controls (290.56,229) and (291,229.44) .. (291,229.98) .. controls (291,230.53) and (290.56,230.97) .. (290.02,230.97) .. controls (289.47,230.97) and (289.03,230.53) .. (289.03,229.98) -- cycle ;
                \draw  [color={rgb, 255:red, 0; green, 0; blue, 0 }  ,draw opacity=1 ][fill={rgb, 255:red, 0; green, 0; blue, 0 }  ,fill opacity=1 ][line width=1.5]  (389.03,229.98) .. controls (389.03,229.44) and (389.47,229) .. (390.02,229) .. controls (390.56,229) and (391,229.44) .. (391,229.98) .. controls (391,230.53) and (390.56,230.97) .. (390.02,230.97) .. controls (389.47,230.97) and (389.03,230.53) .. (389.03,229.98) -- cycle ;
                \draw  [color={rgb, 255:red, 0; green, 0; blue, 0 }  ,draw opacity=1 ][fill={rgb, 255:red, 0; green, 0; blue, 0 }  ,fill opacity=1 ][line width=1.5]  (119.03,319.98) .. controls (119.03,319.44) and (119.47,319) .. (120.02,319) .. controls (120.56,319) and (121,319.44) .. (121,319.98) .. controls (121,320.53) and (120.56,320.97) .. (120.02,320.97) .. controls (119.47,320.97) and (119.03,320.53) .. (119.03,319.98) -- cycle ;
                \draw  [color={rgb, 255:red, 0; green, 0; blue, 0 }  ,draw opacity=1 ][fill={rgb, 255:red, 0; green, 0; blue, 0 }  ,fill opacity=1 ][line width=1.5]  (170.03,319.98) .. controls (170.03,319.44) and (170.47,319) .. (171.02,319) .. controls (171.56,319) and (172,319.44) .. (172,319.98) .. controls (172,320.53) and (171.56,320.97) .. (171.02,320.97) .. controls (170.47,320.97) and (170.03,320.53) .. (170.03,319.98) -- cycle ;
                \draw  [color={rgb, 255:red, 0; green, 0; blue, 0 }  ,draw opacity=1 ][fill={rgb, 255:red, 0; green, 0; blue, 0 }  ,fill opacity=1 ][line width=1.5]  (249.03,319.98) .. controls (249.03,319.44) and (249.47,319) .. (250.02,319) .. controls (250.56,319) and (251,319.44) .. (251,319.98) .. controls (251,320.53) and (250.56,320.97) .. (250.02,320.97) .. controls (249.47,320.97) and (249.03,320.53) .. (249.03,319.98) -- cycle ;
                \draw  [color={rgb, 255:red, 0; green, 0; blue, 0 }  ,draw opacity=1 ][fill={rgb, 255:red, 0; green, 0; blue, 0 }  ,fill opacity=1 ][line width=1.5]  (299.03,319.98) .. controls (299.03,319.44) and (299.47,319) .. (300.02,319) .. controls (300.56,319) and (301,319.44) .. (301,319.98) .. controls (301,320.53) and (300.56,320.97) .. (300.02,320.97) .. controls (299.47,320.97) and (299.03,320.53) .. (299.03,319.98) -- cycle ;
                \draw  [color={rgb, 255:red, 0; green, 0; blue, 0 }  ,draw opacity=1 ][fill={rgb, 255:red, 0; green, 0; blue, 0 }  ,fill opacity=1 ][line width=1.5]  (489.03,319.98) .. controls (489.03,319.44) and (489.47,319) .. (490.02,319) .. controls (490.56,319) and (491,319.44) .. (491,319.98) .. controls (491,320.53) and (490.56,320.97) .. (490.02,320.97) .. controls (489.47,320.97) and (489.03,320.53) .. (489.03,319.98) -- cycle ;
                \draw  [color={rgb, 255:red, 0; green, 0; blue, 0 }  ,draw opacity=1 ][fill={rgb, 255:red, 0; green, 0; blue, 0 }  ,fill opacity=1 ][line width=1.5]  (539.03,319.98) .. controls (539.03,319.44) and (539.47,319) .. (540.02,319) .. controls (540.56,319) and (541,319.44) .. (541,319.98) .. controls (541,320.53) and (540.56,320.97) .. (540.02,320.97) .. controls (539.47,320.97) and (539.03,320.53) .. (539.03,319.98) -- cycle ;
                \draw    (330.02,170) -- (271.43,228.57) ;
                \draw [shift={(270.02,229.98)}, rotate = 315.01] [color={rgb, 255:red, 0; green, 0; blue, 0 }  ][line width=0.75]    (10.93,-3.29) .. controls (6.95,-1.4) and (3.31,-0.3) .. (0,0) .. controls (3.31,0.3) and (6.95,1.4) .. (10.93,3.29)   ;
                \draw    (330.02,170.98) -- (291.14,228.33) ;
                \draw [shift={(290.02,229.98)}, rotate = 304.14] [color={rgb, 255:red, 0; green, 0; blue, 0 }  ][line width=0.75]    (10.93,-3.29) .. controls (6.95,-1.4) and (3.31,-0.3) .. (0,0) .. controls (3.31,0.3) and (6.95,1.4) .. (10.93,3.29)   ;
                \draw    (330.02,170.98) -- (388.59,228.58) ;
                \draw [shift={(390.02,229.98)}, rotate = 224.52] [color={rgb, 255:red, 0; green, 0; blue, 0 }  ][line width=0.75]    (10.93,-3.29) .. controls (6.95,-1.4) and (3.31,-0.3) .. (0,0) .. controls (3.31,0.3) and (6.95,1.4) .. (10.93,3.29)   ;
                \draw    (145.02,360.98) -- (328.13,420.37) ;
                \draw [shift={(330.03,420.98)}, rotate = 197.97] [color={rgb, 255:red, 0; green, 0; blue, 0 }  ][line width=0.75]    (10.93,-3.29) .. controls (6.95,-1.4) and (3.31,-0.3) .. (0,0) .. controls (3.31,0.3) and (6.95,1.4) .. (10.93,3.29)   ;
                \draw    (275.02,359.98) -- (328.66,418.51) ;
                \draw [shift={(330.02,419.98)}, rotate = 227.49] [color={rgb, 255:red, 0; green, 0; blue, 0 }  ][line width=0.75]    (10.93,-3.29) .. controls (6.95,-1.4) and (3.31,-0.3) .. (0,0) .. controls (3.31,0.3) and (6.95,1.4) .. (10.93,3.29)   ;
                \draw    (515.02,359.98) -- (331.92,419.37) ;
                \draw [shift={(330.02,419.98)}, rotate = 342.03] [color={rgb, 255:red, 0; green, 0; blue, 0 }  ][line width=0.75]    (10.93,-3.29) .. controls (6.95,-1.4) and (3.31,-0.3) .. (0,0) .. controls (3.31,0.3) and (6.95,1.4) .. (10.93,3.29)   ;
                \draw    (120.02,319.98) -- (143.97,359.28) ;
                \draw [shift={(145.02,360.98)}, rotate = 238.63] [color={rgb, 255:red, 0; green, 0; blue, 0 }  ][line width=0.75]    (10.93,-3.29) .. controls (6.95,-1.4) and (3.31,-0.3) .. (0,0) .. controls (3.31,0.3) and (6.95,1.4) .. (10.93,3.29)   ;
                \draw    (171.02,319.98) -- (146.09,359.3) ;
                \draw [shift={(145.02,360.98)}, rotate = 302.38] [color={rgb, 255:red, 0; green, 0; blue, 0 }  ][line width=0.75]    (10.93,-3.29) .. controls (6.95,-1.4) and (3.31,-0.3) .. (0,0) .. controls (3.31,0.3) and (6.95,1.4) .. (10.93,3.29)   ;
                \draw    (250.02,318.98) -- (273.97,358.28) ;
                \draw [shift={(275.02,359.98)}, rotate = 238.63] [color={rgb, 255:red, 0; green, 0; blue, 0 }  ][line width=0.75]    (10.93,-3.29) .. controls (6.95,-1.4) and (3.31,-0.3) .. (0,0) .. controls (3.31,0.3) and (6.95,1.4) .. (10.93,3.29)   ;
                \draw    (301.02,318.98) -- (276.09,358.3) ;
                \draw [shift={(275.02,359.98)}, rotate = 302.38] [color={rgb, 255:red, 0; green, 0; blue, 0 }  ][line width=0.75]    (10.93,-3.29) .. controls (6.95,-1.4) and (3.31,-0.3) .. (0,0) .. controls (3.31,0.3) and (6.95,1.4) .. (10.93,3.29)   ;
                \draw    (490.02,318.98) -- (513.97,358.28) ;
                \draw [shift={(515.02,359.98)}, rotate = 238.63] [color={rgb, 255:red, 0; green, 0; blue, 0 }  ][line width=0.75]    (10.93,-3.29) .. controls (6.95,-1.4) and (3.31,-0.3) .. (0,0) .. controls (3.31,0.3) and (6.95,1.4) .. (10.93,3.29)   ;
                \draw    (541.02,318.98) -- (516.09,358.3) ;
                \draw [shift={(515.02,359.98)}, rotate = 302.38] [color={rgb, 255:red, 0; green, 0; blue, 0 }  ][line width=0.75]    (10.93,-3.29) .. controls (6.95,-1.4) and (3.31,-0.3) .. (0,0) .. controls (3.31,0.3) and (6.95,1.4) .. (10.93,3.29)   ;
                \draw  [dash pattern={on 4.5pt off 4.5pt}]  (270.02,228.98) -- (121.72,319.95) ;
                \draw [shift={(120.02,321)}, rotate = 328.47] [color={rgb, 255:red, 0; green, 0; blue, 0 }  ][line width=0.75]    (10.93,-3.29) .. controls (6.95,-1.4) and (3.31,-0.3) .. (0,0) .. controls (3.31,0.3) and (6.95,1.4) .. (10.93,3.29)   ;
                \draw  [dash pattern={on 4.5pt off 4.5pt}]  (290.02,229.98) -- (299.79,317.98) ;
                \draw [shift={(300.02,319.97)}, rotate = 263.66] [color={rgb, 255:red, 0; green, 0; blue, 0 }  ][line width=0.75]    (10.93,-3.29) .. controls (6.95,-1.4) and (3.31,-0.3) .. (0,0) .. controls (3.31,0.3) and (6.95,1.4) .. (10.93,3.29)   ;
                \draw  [dash pattern={on 4.5pt off 4.5pt}]  (390.02,229.98) -- (488.54,319.64) ;
                \draw [shift={(490.02,320.98)}, rotate = 222.3] [color={rgb, 255:red, 0; green, 0; blue, 0 }  ][line width=0.75]    (10.93,-3.29) .. controls (6.95,-1.4) and (3.31,-0.3) .. (0,0) .. controls (3.31,0.3) and (6.95,1.4) .. (10.93,3.29)   ;
                
                \draw    (210,306.8) -- (210,334.8) ;
                \draw    (340,306.8) -- (340,334.8) ;
                \draw    (450,306.8) -- (450,334.8) ;

                \draw  [dash pattern={on 4.5pt off 4.5pt}]  (270.02,229.98) -- (171.52,318.65) ;
                \draw [shift={(170.03,319.98)}, rotate = 318.01] [color={rgb, 255:red, 0; green, 0; blue, 0 }  ][line width=0.75]    (10.93,-3.29) .. controls (6.95,-1.4) and (3.31,-0.3) .. (0,0) .. controls (3.31,0.3) and (6.95,1.4) .. (10.93,3.29)   ;

                \draw  [dash pattern={on 4.5pt off 4.5pt}]  (290.02,229) -- (250.82,318.15) ;
                \draw [shift={(250.02,319.98)}, rotate = 293.73] [color={rgb, 255:red, 0; green, 0; blue, 0 }  ][line width=0.75]    (10.93,-3.29) .. controls (6.95,-1.4) and (3.31,-0.3) .. (0,0) .. controls (3.31,0.3) and (6.95,1.4) .. (10.93,3.29)   ;
                
                \draw (328,229.2) node [anchor=north west][inner sep=0.75pt]    {$\cdots $};
                \draw (138,319.2) node [anchor=north west][inner sep=0.75pt]    {$\cdots $};
                \draw (331.02,156.38) node [anchor=north west][inner sep=0.75pt]  [font=\footnotesize]  {$v_{\text{source}}$};
                \draw (332.02,423.38) node [anchor=north west][inner sep=0.75pt]  [font=\footnotesize]  {$v_{\text{sink}}$};
                \draw (256.02,214.38) node [anchor=north west][inner sep=0.75pt]  [font=\scriptsize]  {$v_{1}^{\prime }$};
                \draw (299.02,214.4) node [anchor=north west][inner sep=0.75pt]  [font=\scriptsize]  {$v_{2}^{\prime }$};
                \draw (391.02,213.4) node [anchor=north west][inner sep=0.75pt]  [font=\scriptsize]  {$v_{s}^{\prime }$};
                \draw (102.02,322.38) node [anchor=north west][inner sep=0.75pt]  [font=\scriptsize]  {$v_{10}$};
                \draw (173.02,322.4) node [anchor=north west][inner sep=0.75pt]  [font=\scriptsize]  {$v_{1\gamma ^{\prime }}$};
                \draw (268,319.2) node [anchor=north west][inner sep=0.75pt]    {$\cdots $};
                \draw (232.02,322.38) node [anchor=north west][inner sep=0.75pt]  [font=\scriptsize]  {$v_{20}$};
                \draw (302.02,322.4) node [anchor=north west][inner sep=0.75pt]  [font=\scriptsize]  {$v_{2\gamma ^{\prime }}$};
                \draw (508,319.2) node [anchor=north west][inner sep=0.75pt]    {$\cdots $};
                \draw (462.02,322.38) node [anchor=north west][inner sep=0.75pt]  [font=\scriptsize]  {$v_{s 0}$};
                \draw (542.02,322.4) node [anchor=north west][inner sep=0.75pt]  [font=\scriptsize]  {$v_{s \gamma ^{\prime }}$};
                \draw (131.02,362.4) node [anchor=north west][inner sep=0.75pt]  [font=\scriptsize]  {$v_{1}$};
                \draw (261.02,362.38) node [anchor=north west][inner sep=0.75pt]  [font=\scriptsize]  {$v_{2}$};
                \draw (517.02,362.4) node [anchor=north west][inner sep=0.75pt]  [font=\scriptsize]  {$v_{s}$};
                \draw (388,319.2) node [anchor=north west][inner sep=0.75pt]    {$\cdots $};
                \draw (388,359.2) node [anchor=north west][inner sep=0.75pt]    {$\cdots $};
            \end{tikzpicture}
            }
        \caption{An overview of the flow network for extendable checking algorithm.} \label{fig:two-clique-mathcing}
    \end{figure}
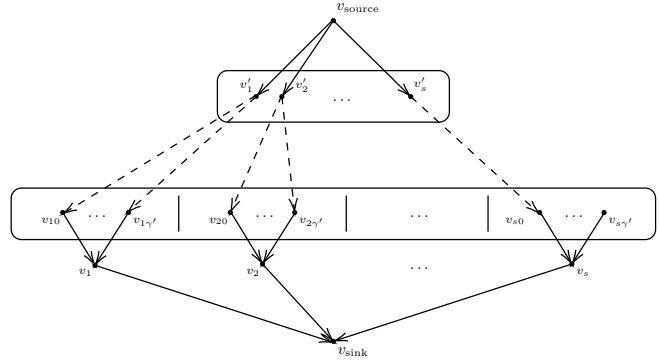

    For each vertex $u_i$ of $S \setminus Q$, consider $\gamma^\prime + 1$ vertices $v_{i 0}, \ldots, v_{i \gamma^\prime}$ that each represents the location of the paired or singleton color class containing $u_i$ between the gaps of color classes $(C^\prime_1, \ldots, C^\prime_{\gamma^\prime})$. The vertices $v_{i \lambda}$ for $1 \leq i \leq s$ and $0 \leq \lambda \leq \gamma^\prime$, form a part of our bipartite subgraph. After considering these main vertices, we create an extra auxiliary vertex $v_i$ for each of the vertices $u_i$ of $S \setminus Q$ to handle our constraints. For each vertex $u_j^\prime$ of $S^\prime \setminus Q$, consider a vertex $v_j^\prime$ that represents $u_j^\prime$ in our bipartite subgraph. In addition, let us add $s - s^\prime$ dummy vertices $v_{s^\prime + 1}^\prime, \ldots, v_s^\prime$ that represent null nodes that will be used to determine singleton color classes. These dummy vertices, along with the vertices $v_j^\prime$ for $1 \leq j \leq s^\prime$, form the second part of our bipartite subgraph. Now it is time to complete the flow network. Figure~\ref{fig:two-clique-mathcing} shows an example of such a flow network. 

    First, we add $(v_{\text{source}}, v_j^\prime)$ edges for all $1 \leq j \leq s$ to complete the first half of our network. Then, we add $(v_{i \lambda}, v_i)$ edges for all $1 \leq i \leq s$ and $0 \leq \lambda \leq \gamma^\prime$. We also add $(v_i, v_{\text{sink}})$ edges for all $1 \leq i \leq s$ to complete the second half of the network. Now, it remains to add the edges between these two halves. For each $1 \leq i, j \leq s$ and $0 \leq \lambda \leq \gamma^\prime$, we add the $(v_j^\prime, v_{i \lambda})$ edge to the network, if and only if creating a paired color class with $u_j^\prime$ and $u_i$ and placing it immediately after the $\lambda^{\text{th}}$ color class of $(C^\prime_1, \ldots, C^\prime_{\gamma^\prime})$ does not violate any constraints. That is, for each $p\in \{1,\ldots,\lambda\}$, $u_j^\prime$ and $u_i$ must have an edge to at least one vertex in $C^\prime_p$. Also, for each $q\in \{\lambda+1,\ldots,\gamma^\prime\}$, and $y\in C^\prime_{q}$, $y$ is adjacent to a vertex in  $\{u_j^\prime,u_i\}$. When $v^\prime_j$ represents a null node, we only add the edge $(v^\prime_j, v_{i\gamma^\prime})$, if and only if the singleton color class consisting of $u_i$ can be placed immediately after the last color class of $(C^\prime_1, \ldots, C^\prime_{\gamma^\prime})$. This holds when $u_i$ has at least one neighbor in each of the color classes in $(C^\prime_1, \ldots, C^\prime_{\gamma^\prime})$.
    
    Finally, the capacity of all edges is set to $1$, and we are ready to use this flow network in our algorithm. Since the capacities are integers, the maximum flow in this network will also be integral. By running a \textsc{Maximum Flow} solver on this network, we obtain a maximum matching in the bipartite subgraph formed by the $s + 1$ main sections, with the condition that for each $1 \leq i \leq s$, only one of the vertices $v_{i 0}, \ldots, v_{i \gamma^\prime}$ can be a part of the matching. This matching represents the paired and singleton color classes placed in their respective gaps. Note that singletons can appear in the final gap, according to the constraints on adding edges in the flow. The order of the pairs inside each gap can be arbitrary, as each pair includes a vertex from each clique. However, in the final gap, where singletons may appear, we arrange the pairs in an arbitrary order before placing the singletons (the order of the singletons can also be arbitrary). As outlined in Lemma~\ref{lemma:singletons-at-the-end-G}, in any arbitrary Grundy coloring, a singleton can be placed at the end without affecting any properties of the Grundy coloring. This is why we only consider singletons to appear at the end.
    The size of this matching determines the extendibility of $(C^\prime_1, \ldots, C^\prime_{\gamma^\prime})$: if it is equal to $s$, the answer is positive; otherwise, it is negative. The total running time of this algorithm is
    $$
        \underbrace{s (\gamma^\prime + 2) + 2}_{\text{Vertices}} + \underbrace{s (\gamma^\prime + 2)}_{\text{Fixed edges}} + \underbrace{s^2 \gamma^\prime}_{\text{Other edges}} \times \underbrace{2 (n - r - s - s^\prime)}_{\substack{\text{Needed operations} \\ \text{for checking an edge}}} + \underbrace{O\Big( \big( s (\gamma^\prime + 2) + 2 \big)^3 \Big)}_{\substack{\text{Solving for} \\ \text{\textsc{Maximum Flow} \cite{10.1145/48014.61051}}}}.
    $$
    This running time can be upper-bounded, yielding a total time complexity of $ O(n^2) + O(n^2) + O(n^3) \times O(n) + O\big((n^2)^3\big) = O(n^6)$. Thus, we provided a $O(n^6)$ algorithm to check the extendability of the Grundy coloring $(C^\prime_1, \ldots, C^\prime_{\gamma^\prime})$ of $G[Q \cup R]$, concluding the proof. \qed
\end{proof}

\begin{theorem} \label{theorem:2-cluster-modulator-fpt}
    There is an algorithm that given a graph $G$ and a 2-cluster modulator $R$ of size $r$, runs in time $O(2^{O(r^2)}n^6)$ and outputs a Grundy coloring of $G$ using the maximum number of colors.
\end{theorem}

\begin{proof}
    Let $F \subseteq V(G)$ be the subset from Lemma~\ref{lemma:R_colorclasses-2}, with $|F| \leq r 2^{r+1}$. Let $G^\prime=G[R\cup F]$. By Lemma~\ref{lemma:R_colorclasses-2}, there exists a subset $Q \subseteq V(G) \setminus R$ of size at most $2r$ such that $G[Q\cup R]$ has the following property: an optimal Grundy coloring $(C_1, \ldots, C_{\gamma})$ of $G$ contains a subsequence $(C^\prime_1, \ldots, C^\prime_{\gamma^\prime})$.
    
    Therefore, we can guess the correct $Q$ and a Grundy coloring $(C^\prime_1, \ldots, C^\prime_{\gamma^\prime})$ of $G[Q\cup R]$, extendable in $2^{O(r^2)}$ time. Testing if $(C^\prime_1, \ldots, C^\prime_{\gamma^\prime})$ is extendable can be done in $O(n^6)$, as per Theorem~\ref{theorem:poly-extendable-2}. After confirming extendability, Lemma~\ref{lemma:grundy-calc-2} allows us to calculate the Grundy number and Grundy coloring of $G$. Using a similar approach explained in Section~\ref{sec:clique}, constructing $F$ takes $O(n+m)$ time. Thus, the total running time would be of $O(2^{O(r^2)}n^6 )$, as desired. \qed
\end{proof}

\section{Solving the $k$-Cluster Modulator Case} \label{sec:kcluster}

Finally, we consider the more general scenario involving a $k$-cluster modulator. In this setting, after removing the $k$-cluster modulator $R$ from $G$, the remaining graph $G - R$ forms a $k$-cluster graph.

Let $K_1, \ldots, K_k$ denote the connected components of $G - R$, each of which is a complete graph. Define $S_1 = V(K_1), \ldots, S_k = V(K_k)$ to be the cliques in $G$. Throughout this section, we denote $r=|R|$.  To address this case, we encode the problem as an Integer Linear Program (ILP) with a number of variables bounded by a function of $r+k$. First, we define some notations and prove some auxiliary lemmas.

\begin{definition}
Let $(C_1, \ldots, C_\gamma)$ be a Grundy coloring of $G$. Define $\text{CL}(C_i)$ to be the set of cliques with a vertex in the color class $C_i$.    
\end{definition}

\begin{lemma} \label{lemma:cliques-fill-in-order}
    Let $(C_1, \ldots, C_\gamma)$ be a Grundy coloring of $G$ and let $(C_{\alpha_1}, \ldots, C_{\alpha_t})$ be the color classes that do not contain a vertex from the modulator $R$, with $\alpha_1 \leq \cdots \leq \alpha_t$. Then we have $\text{CL}(C_{\alpha_t}) \subseteq \cdots \subseteq \text{CL}(C_{\alpha_1})$.
\end{lemma}

\begin{proof}
    Suppose, for contradiction, that there exist two color classes $C_{\alpha_i}$ and $C_{\alpha_j}$ with $i > j$ such that $ \text{CL}(C_{\alpha_i}) \setminus \text{CL}(C_{\alpha_j}) \neq \emptyset$. Let $S_x$ be a clique in $\text{CL}(C_{\alpha_i}) \setminus \text{CL}(C_{\alpha_j})$. By definition, there exists a vertex $v \in S_x$ that appears in $C_{\alpha_i}$. Since $S_x \notin \text{CL}(C_{\alpha_j})$, this vertex $v$ is not adjacent to any vertex in $C_{\alpha_j}$. However, this contradicts the properties of a Grundy coloring, which requires that every vertex in a later color class (here, $C_{\alpha_i}$) must be adjacent to at least one vertex in every earlier color class (here, $C_{\alpha_j}$). This contradiction proves the containment relationship, establishing the desired result. \qed
\end{proof}

Similar to the previous sections, let us define an equivalence relation $\sim_R$ on the set $\bigcup_{i = 1}^k S_i$. For any two vertices $u, v \in \bigcup_{i = 1}^k S_i$, we say $u \sim_R v$ if and only if they have the same closed neighborhood, i.e., $N_G[u] = N_G[v]$. Notice that each equivalence class is a subset of a clique. The equivalence classes of $\sim_R$, that are subsets of the clique $S_i$, are denoted by $E_{i,1}, \ldots, E_{i,q(i)}$, where $q(i)$ is the number of equivalent classes in $S_i$. The total number of equivalence classes is $\sum_{i = 1}^k q(i) \leq k 2^{r}$. Now from each equivalence class $E_{i,j}$ arbitrarily select a subset $F_{i,j}\subseteq E_{i,j}$ of size $\min\{r,|E_{i,j}|\}$. Let $F = \bigcup_{i = 1}^k \Big( \bigcup_{j = 1}^{q(i)} F_{ij} \Big)$. Note that we have $|F|\leq rk2^{r}$.

\begin{lemma} \label{lemma:R_colorclasses-k}
    There is vertex subset $Q\subseteq F$ of size at most $kr$ and a (not necessarily optimal) Grundy coloring $(C^\prime_1, \ldots, C^\prime_{\gamma^\prime})$ of $G[Q\cup R]$ with the following property. There is an optimal Grundy coloring $(C_1, \ldots, C_{\gamma})$ of $G$ such that $(C^\prime_1, \ldots, C^\prime_{\gamma^\prime})$ is a subsequence of $(C_1, \ldots, C_{\gamma})$.
\end{lemma}

\begin{proof}
    The proof is almost identical to the proof of Lemma~\ref{lemma:R_colorclasses-2}. \qed
\end{proof}

Recall the definition of extendable Grundy coloring of $G[R\cup Q]$ (Definition~\ref{def:ext}). 

\begin{lemma} \label{lemma:grundy-calc-k}
    Let $(C^\prime_1, \ldots, C^\prime_{\gamma^\prime})$ be a (not necessarily optimal) Grundy coloring of $G[Q\cup R]$, that is extendable. Let $\beta= \max_{i = 1}^k\{|S_i\setminus Q|\}$. Then, Grundy number of $G$ is $\gamma^\prime+\beta$.
\end{lemma}

\begin{proof}
    By applying Lemma~\ref{lemma:cliques-fill-in-order} instead of Lemma~\ref{lemma:singletons-one-clique} within the detailed steps of the proof of Lemma~\ref{lemma:grundy-calc-2}, it is clear that the proof is derived similarly. \qed
\end{proof}

Before stating the main theorems of this section, we need to mention a useful theorem for further use.

\begin{theorem}[\cite{cygan2015parameterized}] \label{ilp-fpt}
    An \textsc{Integer Linear Programming Feasibility} instance of size $L$ with $p$ variables can be solved using $O(p^{2.5p+o(p)}L)$ arithmetic operations and space polynomial in $L$.
\end{theorem}

\begin{theorem} \label{theorem:poly-extendable-k}
    There is an algorithm that given a graph $G$ and a Grundy coloring $(C^\prime_1, \ldots, C^\prime_{\gamma^\prime})$ of $G[Q\cup R]$, decides if $(C^\prime_1, \ldots, C^\prime_{\gamma^\prime})$ is extendable. This algorithm runs in time $O(p^{2.5p+o(p)})$, where $p = O(2^{kr} r)$.
\end{theorem}

\begin{proof}
    The objective is to construct new color classes and add them to the existing Grundy coloring $(C^\prime_1, \ldots, C^\prime_{\gamma^\prime})$ of $G[Q \cup R]$. There are $\gamma^\prime + 1$ possible positions where a new color class can be inserted, referred to as gaps, denoted $G_0, \ldots, G_{\gamma^\prime}$, between the existing color classes $C^\prime_1, \ldots, C_{\gamma^\prime}$. Each new color class to be added consists solely of vertices from the cliques and can be represented by a tuple $(s_1, \ldots, s_k)$, where each $s_i$ indicates which equivalence class from the clique $S_i$ is chosen for the color class. This representation is sufficient because vertices within the same equivalence class can be swapped, due to Lemma~\ref{lemma:swap-vertices}.
    
    The integer linear program that models this problem would be as below:
    \begin{alignat*}{3}
        & \text{maximize} & 0& \\
        & \text{subject to} \qquad& &\sum_{j = 1}^{\gamma^\prime} \sum_{l: \exists X_{lj}} X_{lj} = \beta, \\
        && &\sum_{j = 1}^{\gamma^\prime} \sum_{l: T_l \cap E \neq \phi} X_{lj} = |E|, \qquad & \forall & \text{ equivalent class } E \\
        && &Y_{i j} \leq \hspace{15pt} \sum_{l: T_l \cap S_i \neq \phi} X_{lj}, \qquad & i &\in \{1, \ldots, k\},~j \in \{0, \ldots, \gamma^\prime\} \\
        && &Y^\prime_{i j} \geq 1 - \sum_{l: T_l \cap S_i \neq \phi} X_{lj}, \qquad & i &\in \{1, \ldots, k\},~j \in \{0, \ldots, \gamma^\prime\} \\
        && &Y_{i j} + Y^\prime_{i j} = 1, \qquad & i &\in \{1, \ldots, k\},~j \in \{0, \ldots, \gamma^\prime\} \\
        && &Y_{i (j+1)} \geq Y_{i j}, \qquad & i &\in \{1, \ldots, k\},~j \in \{1, \ldots, \gamma^\prime\} \\
        && &Y^\prime_{i j}\leq Y^\prime_{i (j+1)}, \qquad & i &\in \{1, \ldots, k\},~j \in \{1, \ldots, \gamma^\prime\} \\
        && &Y_{i j} \in \{0,1\},\qquad & i &\in \{1, \ldots, k\},~j \in \{0, \ldots, \gamma^\prime\} \\
        && &Y^\prime_{i j} \in \{0,1\},\qquad & i &\in \{1, \ldots, k\},~j \in \{0, \ldots, \gamma^\prime\} \\
        && &X_{l j} \in \mathbb{N} \cup 0,\qquad & \exists & X_{l j},~j \in \{0, \ldots, \gamma^\prime\} \\
    \end{alignat*}

    Where the variables $X_{lj}$, $Y_{ij}$ and $Y^\prime_{ij}$ are defined as
    
    \begin{align*}
        X_{lj} &= \text{the number of appearances of tuple color class } T_l \text{ in gap } G_j \text{, if feasible}, \\
        Y_{i j} &= \text{boolean variable indicating that some vertex from } S_i \text{ appears in gap } G_j, \\
        Y^\prime_{i j} &= \text{boolean variable indicating that no vertex from } S_i \text{ appears in gap } G_j.
    \end{align*}
    
    The notation $\exists X_{lj}$ implies that $X_{lj}$ is a feasible variable, i.e. it has passed a preprocessing step, verifying whether assigning the tuple color class $T_l$ to gap $G_j$ violates the Grundy coloring constraints or not. We suggest that $(C^\prime_1, \ldots, C^\prime_{\gamma^\prime})$ is extendable if and only if the above integer linear program has a solution.

    First, assume there is a valid way to construct the remaining color classes. These classes will satisfy two key properties: the Grundy coloring property and the quantity property. The Grundy coloring property is captured by constraints 3, 4, 5, 6, and 7, while the quantity property is captured by constraints 1 and 2. The remaining constraints are trivially satisfied, yielding a feasible solution to the integer linear program.

    Conversely, assume the integer linear program has a feasible solution. We show this solution provides a valid construction for the remaining color classes. By Lemma~\ref{lemma:grundy-calc-k}, constraint 1 ensures that the total number of color classes does not exceed the number of vertices in the cliques, while constraint 2 ensures each equivalence class is used the correct number of times. The preprocessing step ensures that infeasible assignments to gaps are excluded, so $\nexists X_{lj}$ for invalid assignments. Constraints 3, 4, 5, 6, and 7 enforce the Grundy coloring property by maintaining the required order between gaps, as explained in Lemma~\ref{lemma:cliques-fill-in-order}. Using the result of Lemma~\ref{lemma:swap-vertices}, any incorrect assignments within a gap can be corrected by swapping vertices. Thus, all conditions are satisfied, and the color classes can be constructed as required.
    
    Therefore, the above integer linear program operates as intended and can be used to solve the problem. Using the parameters $p = 2^{kr} (r+1) + 2 k (r+1) = O(2^{kr} r)$ and $L = \big(2^r k + 5 k (r+1) + p\big) p = O(p^2)$, and based on Theorem~\ref{ilp-fpt}, we find the total running time of the algorithm is
    $$
        \underbrace{O(2^{kr} r)}_{\text{Preprocessing Phase}} + \hspace{15pt} \underbrace{O(p^{2.5p+o(p)} p^2)}_{\text{ILP Phase}}.
    $$
    This results in a time complexity of $O(p^{2.5p+o(p)})$, completing the proof. \qed
\end{proof}

\begin{theorem} \label{theorem:k-cluser-modulator-fpt}
    There is an algorithm that given a graph $G$ and a $k$-cluster modulator $R$ of size $r$, outputs a Grundy coloring of $G$ using the maximum number of colors. This algorithm runs in time $O(2^{O(kr^2)} p^{2.5p+o(p)})$, where $p = O(2^{kr} r)$.
\end{theorem}

\begin{proof}
    Let $F \subseteq V(G)$ be the subset from Lemma~\ref{lemma:R_colorclasses-k}. Following a similar approach as in the proof of Theorem~\ref{theorem:2-cluster-modulator-fpt}, we find that $|F| \leq r 2^{r+k}$. We can guess the correct $Q$ and a Grundy coloring $(C^\prime_1, \ldots, C^\prime_{\gamma^\prime})$ of $G[Q\cup R]$, extendable in $2^{O(kr^2)}$ time. To check whether $(C^\prime_1, \ldots, C^\prime_{\gamma^\prime})$ we can use the algorithm of Theorem~\ref{theorem:poly-extendable-k} in $O(p^{2.5p+o(p)})$ time. By applying the same method as before, we obtain a total running time of $O(2^{O(kr^2)} p^{2.5p+o(p)})$, as desired. \qed
\end{proof}

\section*{Acknowledgments}

We would like to sincerely thank the anonymous reviewers at CALDAM 2025 for their insightful and constructive comments. Their feedback significantly contributed to improving the clarity and quality of this work.

\bibliographystyle{splncs04}
\bibliography{ref.bib}

\end{document}